\documentclass[11pt]{article}
\usepackage[margin=1in]{geometry}
\usepackage{times}
\usepackage{amsmath,amsthm,amssymb,amsfonts}
\usepackage{wrapfig,color}
\usepackage{mathrsfs,upgreek}
\usepackage[colorlinks,linkcolor=black,bookmarksopen,
bookmarksnumbered,citecolor=black,urlcolor=black]{hyperref}
\usepackage{graphicx} 

\usepackage{float}

\newcommand{\R}{ {\mathbb R} }
\newcommand{\Z}{ {\mathbb Z} }

\newcommand{\vb}{ \mathbf{b} }
\newcommand{\vc}{ \mathbf{c} }
\newcommand{\vf}{ \mathbf{f} }
\newcommand{\vx}{ \mathbf{x} }
\newcommand{\vzero}{ \mathbf{0} }
\newcommand{\St}{\operatorname{St}}
\newcommand{\Lk}{\operatorname{Lk}}
\newcommand{\Dom}{\operatorname{Dom}}
\newcommand{\crt}{\mathscr{C}}
\newcommand{\crtM}{\mathscr{C}_M}
\newcommand{\crtL}{\mathscr{C}_L}
\newcommand{\subG}{\mathscr{S}}

\DeclareMathOperator{\boundary}{\partial}
\newcommand{\homo}{{\sf H}}
\newcommand{\chains}{{\sf C}}
\newcommand{\cycles}{{\sf Z}}
\newcommand{\bdys}{{\sf B}}

\DeclareMathOperator{\image}{{\mathrm im\,}}
\newcommand{\rltvbndry}[3]{\operatorname{\boundary}_{#1}^{\, (#2,#3)}}

\definecolor{darkgrn}{rgb}{0, 0.8, 0}

\newtheorem{theorem}{Theorem}[section]

\newtheorem{corollary}[theorem]{Corollary}
\newtheorem{proposition}[theorem]{Proposition}
\theoremstyle{definition}
\newtheorem{definition}[theorem]{Definition}
\theoremstyle{remark}

\begin{document}

\title{Edge Contractions and Simplicial Homology}

\author{
Tamal K. Dey\thanks{
Department of Computer Science and Engineering,
The Ohio State University, Columbus, OH 43210, USA.
Email: {\tt tamaldey@cse.ohio-state.edu}}
\quad\quad
Anil N. Hirani\thanks{
Department of Computer Science, 
University of Illinois at Urbana-Champaign, Urbana, IL 61801, USA.
Email: {\tt hirani@illinois.edu}}
\quad\quad Bala Krishnamoorthy\thanks{
Department of Mathematics,
Washington State University, Pullman, WA, USA.
Email: {\tt bkrishna@math.wsu.edu}}
\quad\quad Gavin Smith\thanks{
Department of Mathematics,
Washington State University, Pullman, WA, USA.
Email: {\tt gsmith@math.wsu.edu}} 
}


\maketitle
\thispagestyle{empty}
\begin{abstract}
  We study the effect of edge contractions on simplicial
  homology because these contractions have turned to be useful in various
  applications involving topology. It was observed previously that 
  contracting edges that
  satisfy the so called {\em link condition} preserves homeomorphism
  in low dimensional complexes, and homotopy in general.  
  But, checking the link condition involves computation in all dimensions,
  and hence can be costly, especially in high dimensional
  complexes. We define a weaker and more local condition called the
  $p$-{\em link} condition for each dimension $p$, and study its
  effect on edge contractions. We prove the following: (i) For homology
  groups, edges satisfying the
  $p$- and $(p-1)$-link conditions can be contracted without
  disturbing the $p$-dimensional homology group. (ii) For
  relative homology groups, the $(p-1)$-, and the $(p-2)$-link conditions 
  suffice to guarantee that the contraction does not introduce any new class in 
  any of the resulting relative
  homology groups, though some of the existing classes can be
  destroyed. Unfortunately, the surjection in relative homolgy groups
  does not guarantee that no new relative
  torsion is created. (iii) For torsions, edges satisfying the 
  $p$-link condition alone
  can be contracted without creating any new relative torsion and
  the $p$-link condition cannot be avoided.  The
  results on relative homology and relative torsion are motivated by
  recent results on computing optimal homologous chains, which state
  that such problems can be solved by linear programming if the
  complex has no relative torsion.  Edge contractions that do not
  introduce new relative torsions, can safely be availed in these contexts.
\end{abstract}

\section{Introduction}
The study of edge contractions in the context of
graph theory~\cite{WB04}, especially in graph minor theory~\cite{MT01}
have resulted into many beautiful results.
The extension of edge contractions to simplicial complexes 
where the structure not
only has vertices and edges, but also higher dimensional simplices
has also turned out to be beneficial for shape representation
in graphics and visualization~\cite{HDDMS1993,GaHe1997} and 
recently in topological
data analysis~\cite{ALS13}. In this paper, we present several results
about edge contractions in simplicial complexes that can benefit 
extraction of topological attributes from a shape or data representation.

Topological attributes such as the rank of the homology groups, also
known as betti numbers, and cycles representing the homology classes
carry important information about a shape. Naturally, efforts have
ensued to compute them efficiently in various applications. Examples
include computing topological features for low dimensional complexes
in graphics, visualizations, and sensor
networks~\cite{Bias2012,SiGh2007,TaJa2009,WHDS2004}, and for higher
dimensional ones in data analysis~\cite{Gh2008barcodes}.  As the input
sizes in these applications grow with the advances in data generation
technology, methods to speed up the homology computations become more
demanding. For example, the Vietoris-Rips complex has been recognized
as a versatile data structure for inferring homological attributes
from point cloud data~\cite{Gh2008barcodes}. However, because of its inclusive
nature, this complex tends to contain a large number of simplices and
the computation becomes prohibitive when the input point set is large.
One way to tackle this issue is to use {\em edge
  contractions}~\cite{DEGN1999,GaHe1997} that contract edges and
collapse other simplices as a result. Quite naturally, edge
contractions have already been proposed to tame the size of
large Rips complexes~\cite{ALS13}. In a recent work on 
computing optimal homologous cycles (OHCP), we have shown how
relative homology and torsions play a role in guaranteeing a polynomial
time optimization~\cite{DeHiKr2011}. 
To avail the benefit of edge contractions  in this
context we need to understand its effect on relative homology and torsion. 
Motivated by these applications, we make a systematic
study of the effects of edge contractions on simplicial homology.

The effects of edge contractions on topology was initially studied by
Walkup~\cite{Walkup70} for $3$-manifolds
and then by Dey et al.~\cite{DEGN1999} for more general domains. 
They showed that an edge $e$ in a
$2$-complex or a $3$-manifold can be contracted while preserving a
homeomorphism between the complexes before and after the edge
contraction if $e$ satisfies a local condition called the {\em link
  condition}.  The condition, due to its locality, is easily checkable
at least for low dimensional complexes. Attali, Lieutier, and Salinas
\cite{ALS2011Blockers} showed that the result can be extended to the
entire class of finite simplicial complexes if only homotopy instead
of homeomorphism needs to be preserved.  Since homotopy equivalent
spaces have isomorphic homology groups, link conditions also suffice
to preserve homology groups. However, verification of the link
condition requires checking it at every dimension, which becomes
costly for higher dimensional simplicial complexes.

In this work, we extend the above results in two directions: (i) we
study edge contractions not only for homology groups, but also for
relative homology groups and torsion groups, and (ii) we define a
weaker and even more local condition called the {\em $p$-link
  condition} for each dimension $p$ and analyze its effect on
homology, relative homology, and torsion. 
Specifically, we prove that an edge
contraction cannot destroy any homology class in dimension $p$ if the
edge satisfies the $p$-link condition alone. Furthermore, no new class
is introduced if the edge satisfies the $(p-1)$-link condition.  For
relative homology, we show that no new homology class in dimension $p$
relative to a particular subcomplex is created if the edge satisfies the
$(p-1)$- and the $(p-2)$-link conditions in that {\em subcomplex}. 
This result also implies that no new relative homology classes 
are generated in the contracted complex as long as the contracting
edge satisfies the $(p-1)$- and $(p-2)$-link conditions in the
{\em original complex}. Unfortunately, this surjectivity in relative homology
classes does not mean surjectivity in {\em relative torsions}, that is,
the torsion subgroups of relative homology groups. Of course,
if the edge additionally satisfies the $p$-link condition in the 
subcomplex, we have isomorphisms in relative homology
groups which necessarily mean that no new relative torsion
is created. 
We strengthen the condition to require the $p$-link condition
alone if only relative torsion is of interest, which is the case for
OHCP. An example shows that one cannot avoid the $p$-link condition if one is 
to guarantee that no new relative torsion is introduced. 

Our result on homology preservation under $p$- and $(p-1)$-link 
conditions can be
used to compute the betti numbers of large complexes after edge
contractions. They can also be availed to compute actual
representative cycles in a small contracted complex, which can then be
pulled back to the original complex through a systematic reversal of
the contractions. Computations are saved by checking the link
conditions for only a few dimensions instead of all dimensions.  Similarly,
our result for relative homology can be used to compute the betti 
numbers of the quotient complexes formed by the original complex relative
to a subcomplex.
Our result on torsion can be availed for computing
the shortest cycle in a given homology class. This problem, termed the
OHCP, is known to be NP-hard in general~\cite{ChFr2010}.  It has been
recently shown that the OHCP is solvable in polynomial time when the
homology is defined over $\Z$ and the simplicial complex does not have
relative torsion~\cite{DeHiKr2011}. Similar results hold for related
problems such as the optimal bounding chain problem (OBCP)
\cite{DuHi2011} and the multiscale simplicial flat norm (MSFN) problem
\cite{IKV2011}. Therefore, edge contractions that preserve the absence
of, or even better, eliminate relative torsion should be preferred.
Figure~\ref{fig-contr_ex} illustrates a case where an edge contraction
does eliminate relative torsion and thus allows computing an optimal
cycle by linear programming as shown in~\cite{DeHiKr2011}.  We use
methods from algebraic topology to prove most of our results.
Interestingly enough, the result on relative torsion requiring the
$p$-link condition alone is proved only with graph theoretic
techniques.

\section{Background} \label{sec-backgr}

We recall some basic concepts and definitions from algebraic topology
relevant to our presentation. Refer to standard books, e.g., ones by
Munkres~\cite{Munkres1984}
for details.

\begin{definition}
  Given a vertex set $V$, a \emph{simplicial complex} $K=K(V)$ is a
  collection of subsets $\{\sigma \subseteq V\}$ where
  $\sigma'\subseteq \sigma$ is in $K$ if $\sigma\in K$.  A subset
  $\sigma\in K$ of cardinality $p+1$ is called a \emph{$p$-simplex}.
  If $\sigma'\subseteq \sigma$ ($\sigma'\subset \sigma$), we call
  $\sigma'$ a \emph{face} (\emph{proper face}) of $\sigma$, and
  $\sigma$ a \emph{coface} (\emph{proper coface}) of $\sigma'$.
  A map $h: K\rightarrow K'$ between two complexes is called
  \emph{simplicial} if for every simplex $\{v_1,v_2,\ldots, v_k\}$ in $K$
  $\{h(v_1),h(v_2),\ldots,h(v_k)\}$ is a simplex in $K'$.
\end{definition}

An oriented simplex $\sigma=\{v_0,v_1,\cdots,v_p\}$, also written as
$v_0v_1\cdots v_p$, is an ordered set of vertices.  The simplices
$\sigma_i$ with coefficients $\alpha_i$ in $\mathbb{Z}$ can be added
formally creating a chain $c = \Sigma_i \alpha_i \sigma_i$.  These
chains form the chain group $\chains_p$.  The boundary $\partial_p
\sigma$ of a $p$-simplex $\sigma$, $p\geq 0$, is the $(p-1)$-chain
that adds all the $(p-1)$-faces of $\sigma$ with orientation taken
into consideration.  This defines a boundary homomorphism $\partial_p:
\chains_p\rightarrow \chains_{p-1}$.  The kernel of $\partial_p$ forms
the $p$-cycle group $\cycles_p(K)$ and its image forms the
$(p-1)$-boundary group $\bdys_{p-1}(K)$.  The homology group
$\homo_p(K)$ is the quotient group $\cycles_p(K)/\bdys_p(K)$.
Intuitively, a $p$-cycle is a collection of oriented $p$-simplices
whose boundary is zero. It is a non-trivial cycle in $\homo_p$, if it
is not a boundary of a $(p+1)$-chain.

For a finite simplicial complex $K$, the groups of chains
$\chains_p(K)$, cycles $\cycles_p(K)$, and $\homo_p(K)$ are all
finitely generated abelian groups. By the fundamental theorem of
finitely generated abelian groups \cite[page 24]{Munkres1984} any such
group $G$ can be written as a direct sum of two groups $G=F \oplus T$
where $F\cong (\Z \oplus\cdots\oplus \Z)$ and
$T\cong(\Z/t_1\oplus\cdots\oplus \Z/t_k)$ with $t_i>1$ and $t_i$
dividing $t_{i+1}$. The subgroup $T$ is called the \emph{torsion} of
$G$. If $T=0$, we say $G$ is \emph{torsion-free}.

A simplicial map $h: K\rightarrow K'$ between two simplicial complexes
induces a homomorphism 
between their homology groups which we write as
$h_*: \homo_p(K)\rightarrow \homo_p(K')$.
The edge contractions that we deal with define such
homomorphisms whose properties in relation to various
$p$-link conditions are the focus of our work.  

\subsection{Link conditions}
\label{sec-edgcontr}
We first define edge contractions formally.
\begin{definition}
  Let $ab=\{a,b\}$ be an edge in a simplicial complex $K$. An edge
  contraction of $K$ is a surjective simplicial map $\gamma_{ab}: K
  \rightarrow K'$ induced by the vertex map $h: V(K) \rightarrow
  V(K')$ where $h$ is identity everywhere except at $b$ for which
  $h(b)=a$.
\end{definition}
 
Authors of~\cite{DEGN1999} investigated when such an edge contraction
results in any change in the topology of the simplicial complex. They
provided a sufficient condition termed the link condition, which
guarantees that topology is preserved for certain simplicial
complexes. This condition has been studied further by Attali, Lieutier, and
Salinas \cite{ALS2011Blockers} recently.

\begin{definition}
  The \emph{star} of a set $X\subseteq K$, denoted $\St X$, is the set
  of cofaces of all $\sigma\in X$. For a subset $S$ of $K$, the
  \emph{closure} of $S$, denoted $\overline{S}$, is the set of
  simplices in $S$ and all of their faces.  Then the \emph{link} of
  $X$, denoted $\Lk X$, is the set of simplices in $\overline{\St X}$
  that do not belong to $\St \overline{X}$.  In the left complex of
  Figure~\ref{fig:plink}, the star of the edge $ab$ consists of $ab,
  abd, abe, abde$. Its link is $d, e, de$.
\end{definition}

\noindent The following definition of the link condition is taken from
the work in~\cite{DEGN1999}, and we introduce a weaker condition
called the $p$-link condition.
\begin{definition}
  \label{def-linkcond}
  An edge $ab\in K$ satisfies the \emph{link condition} in $K$ if and only if
  $\Lk a \cap \Lk b = \Lk ab$.  It satisfies the $p$-link condition in $K$,
  if and only if either (i) $p\leq 0$, or (ii) $p>0$ and every
  $(p-1)$-simplex $\xi$ $\in \Lk a \cap \Lk b$ is also in $\Lk ab$.
\end{definition}

\noindent The $p$-link conditions are weaker than the link condition
in the following sense.

\begin{proposition}
  \label{lem-plink_link_condition}
  For any edge $ab$ in a simplicial complex $K$, $ab$ satisfies the
  link condition if and only if it satisfies the $p$-link condition
  for all $p \leq \dim(K)$.
\end{proposition}
\begin{proof} \label{pf-lem-plink_link_condition}
We prove both directions by contrapositive.  Assume $ab$ does not
satisfy the link condition.  By the definition of star, $\St ab = \St
a \cap \St b$.  Therefore, $\overline{\St ab} \subset \overline{\St a}
\cap \overline{\St b}$, and $\Lk ab \subset \Lk a \cap \Lk b$.
Therefore, there must be a simplex $\xi \in \Lk a \cap \Lk b$ where
$\xi \notin \Lk ab$. Let $p = \dim(\xi) + 1$. Then, $ab$ does not
satisfy the $p$-link condition.  By the definition of link, there must
be a $p$-simplex $\tau_1 \in K$ where $\tau_1 = \xi \cup {a}$, and
also a $p$-simplex $\tau_2 \in K$ where $\tau_2 = \xi \cup {b}$.
Therefore, $p \leq \dim(K)$.

Now assume $ab$ does not satisfy the $p$-link condition for some $p
\leq \dim(K)$.  Then there is some $(p-1)$-simplex $\xi \in \Lk a \cap
\Lk b$ that is not in $\Lk ab$.  Therefore, $ab$ does not satisfy the
link condition.
\end{proof}

\begin{wrapfigure}{r}{3in}
\centering
\vspace*{0.15in}
\includegraphics[scale=0.8]{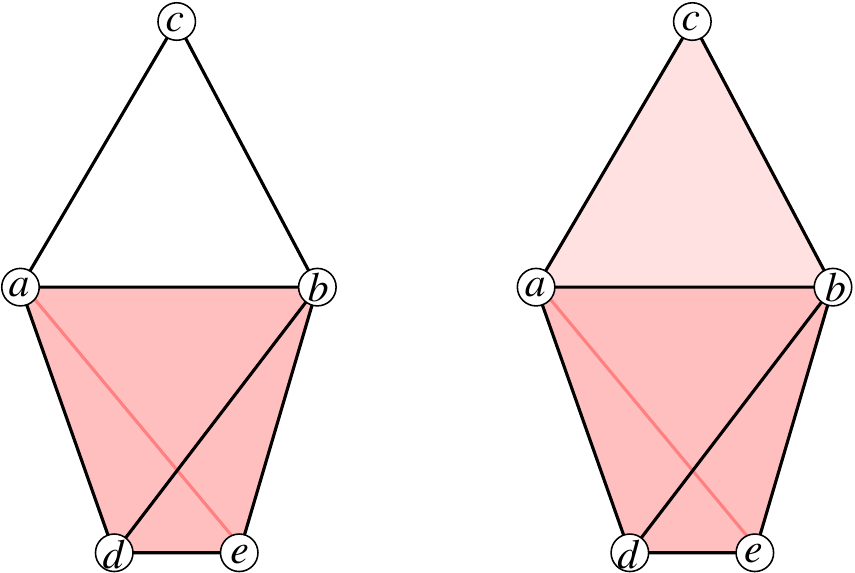}
\caption{$p$-link conditions for $p=1,2$.}
\label{fig:plink}
\vspace*{-0.15in}
\end{wrapfigure}
\noindent As an example, consider the two complexes in
Figure~\ref{fig:plink}.  Both complexes have the lower tetrahedron
adjoining the edge $ab$. The complex on the left has the top triangle
missing. The link of the edge $ab$ contains $d, e, de$ in the left
complex, and it satisfies the $2$-link condition but not the $1$-link
condition. $ab$ satisfies both the $1$-link and the $2$-link
conditions in the complex on the right, since the link of $ab$
contains $c, d, e,$ and $de$.

\section{Homology Preservation}
In this section we study how the $p$-link conditions for edges
being contracted affect the homology groups.  Not surprisingly, the
homology classes in dimension $p$ are maintained intact if the $p$-
and $(p-1)$-link conditions hold.  Specifically, we prove that the
$p$-link condition alone implies that no homology class in
$\homo_p(K)$ is destroyed (injectivity) and the $(p-1)$-link condition
alone implies that no new homology class is created (surjectivity).

\begin{theorem}
  Let $ab$ be an edge in a simplicial complex $K$ and $\gamma_{ab}: K
  \rightarrow K'$ be an edge contraction. Then the induced
  homomorphism $\gamma_{ab*}$ at the homology level has the following
  properties:
  \begin{enumerate}
  \item $\gamma_{ab*}: \homo_p(K) \rightarrow \homo_p(K')$ is
    injective if $ab$ satisfies the $p$-link condition.
  \item $\gamma_{ab*}$ is surjective if $ab$ satisfies the
    $(p-1)$-link condition.
  \end{enumerate}
\label{homo-thm}
\end{theorem}

To prove the above theorem, we use an intermediate complex $\hat{K}$
constructed as follows. Let the cone $v*\sigma$ for a vertex $v\in K$
and a simplex $\sigma\in K$ be defined as the closure of the simplex
$\sigma\cup \{v\}$. The cone $v*T$ for a subcomplex $T\subseteq K$ is
the complex $v*T=\{v*\sigma\,|\, \sigma\in T\}$.  We construct
$\hat{K}= K \cup (a* \overline{\St b})$. In words, $\hat{K}$ is
constructed out of $K$ by adding simplices that are obtained by coning
from $a$ to the closed star of $b$. First, observe that the
$a*\overline{\St b}$ in $\hat{K}$ can deformation retract to $a*\Lk b$
taking $b$ to $a$, and we get $K'$. Therefore, we have the sequence
$K\stackrel{i}{\hookrightarrow} \hat{K}\stackrel{r}{\rightarrow} K'$
where $i$ and $r$ are an inclusion and a deformation retraction,
respectively, and $\gamma_{ab}= r\circ i$; see
Figure~\ref{inc-retract} for an illustration.  At the homology level
we have the sequence of two homomorphisms where the one on right is an
isomorphism induced by a deformation retract.
$$
\homo_p(K)\stackrel{i_*}{\rightarrow}\homo_p(\hat{K})
\stackrel{r_*}{\rightarrow} \homo_p(K').
$$
Since $\gamma_{ab*}=r_*\circ i_*$ at the homology level and $r_*$ is
an isomorphism, we have that $\gamma_{ab*}$ is injective or surjective
if and only if $i_*$ is.
\begin{figure}[hb!]
\centering
\includegraphics[scale=0.67]{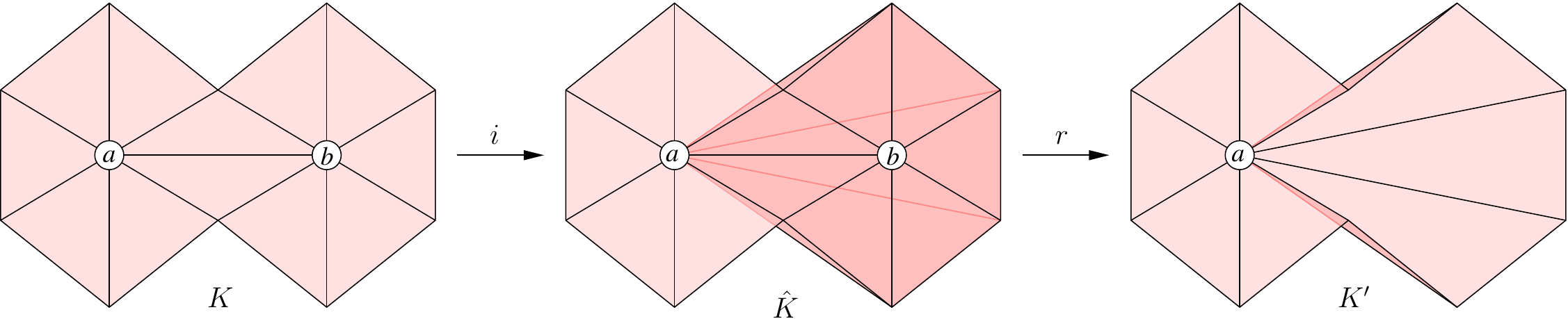}
\caption{Edge contraction as a composition of an inclusion and a retraction.}
\label{inc-retract}
\end{figure}

\begin{proposition}
  $i_*: \homo_p(K)\rightarrow \homo_p(\hat{K})$ is injective if $ab$
  satisfies the $p$-link condition, and is surjective if $ab$
  satisfies the $(p-1)$-link condition.
\label{homology-pres}
\end{proposition}
\begin{proof}
  Since $p$-dimensional homology is determined only by the skeleton up
  to dimension $p+1$, assume that $K$ is only $(p+1)$-dimensional.
  Let $S$ denote the subcomplex $a* \overline{\St b}$. The
  Mayer-Vietoris sequence
  $$
  \homo_p(K\cap S) \stackrel{(j_*,k_*)}{\rightarrow} \homo_p(K)\oplus
  \homo_p(S) \stackrel{(i_*-\ell_*)}{\rightarrow}
  \homo_p(\hat{K})\stackrel{\pi}{\rightarrow} \homo_{p-1}(K\cap S)
  $$
  is exact where $i,j,k,\ell$ are respective inclusion maps and $\pi$
  is the connecting homomorphism~\cite{Hatcher2002}.  We examine the
  maps $j_*,k_*,\ell_*$ and prove the required properties for $i_*$.

  First consider that $ab$ satisfies the $p$-link condition. We
  examine the the complex $K\cap S$. It is the union of $\overline{\St
    b}$ and the added simplices that are cones of $a$ to the simplices
  in $\Lk a \cap \Lk b$. None of these added simplices can be
  $p$-dimensional since otherwise the $p$-link condition is violated.
  Any closed star has trivial homology, and thus
  $\homo_p(\overline{\St b})=0$.  Since no simplex of dimension $p$ or
  more is added to create $K\cap S$ from $\overline{\St b}$, we still
  have $\homo_p(K\cap S)=0$.  Therefore, $\image j_*$ and $\image k_*$
  are trivial. Because of the exactness of the sequence, the map
  $i_*-\ell_*$ is injective. However, $\ell_*$ is a zero map since
  $S=a*\overline{\St b}\,$ has $\homo_p(S)=0$.  It follows that $i_*$
  is injective as we are required to prove.

  Next, consider that $ab$ satisfies the $(p-1)$-link condition. By
  the same logic as above, there is no $(p-1)$-simplex that can be
  added to $\overline{\St b}$ to create $K\cap S$. Therefore,
  $\homo_{p-1}(K\cap S)=0$ which implies that $\image \pi$ is trivial.
  Because of the exactness, we have that $(i_*-\ell_*)$ is surjective
  which implies $i_*$ is surjective as $\ell_*$ is a zero map.
\end{proof}


\section{Relative Homology Preservation}
In this section we study how the $p$-link conditions affect the
relative homology groups.  The motivation comes from a recent result
on the optimal homologous cycle problem (OHCP) whose efficient
solution depends on relative torsion and hence on relative homology
classes \cite{DeHiKr2011}.

First, we set up some background.  Let $L_0$ be a subcomplex of a
simplicial complex $L$. The quotient group
$\chains_p(L)/\chains_p(L_0)$ is called the group of \emph{relative
  $p$-chains} of $L$ modulo $L_0$ and is denoted $\chains_p(L,L_0)$.
The boundary operator $\boundary_p \colon \chains_p(L)\rightarrow
\chains_{p-1}(L)$ and its restriction to $L_0$ induce a homomorphism
\[
\rltvbndry{p}{L}{L_0} \colon \chains_p(L,L_0) \rightarrow
\chains_{p-1}(L,L_0)\, .
\]
Writing $\cycles_p(L,L_0)={\rm ker}\rltvbndry{p}{L}{L_0}$ for
\emph{relative cycles} and $\bdys_{p-1}(L,L_0)=\image
\rltvbndry{p}{L}{L_0}$ for \emph{relative boundaries}, we obtain the
\emph{relative homology group} $\homo_p(L,L_0) =
\cycles_p(L,L_0)/\bdys_p(L,L_0)$.

A {\em pure simplicial complex} of dimension $p$ is formed by a
collection of $p$-simplices and their proper faces.  We consider
relative homology groups $\homo_p(L,L_0)$ where $L \subseteq K$ and
$L_0 \subset L$ are pure subcomplexes of dimensions $(p+1)$ and
$p$, respectively.

To study the effect of edge contraction on relative homology groups,
we note that the simplicial map $\gamma_{ab}: K \rightarrow K'$
naturally extends to pairs as $\gamma_{ab}: (L,L_0)\rightarrow
(L',L_0')$ where $\gamma_{ab}(L)=L'$ and $\gamma_{ab}(L_0)=L_0'$.
Notice that since $L$ is a subcomplex of $K$, it may not satisfy some
link conditions even if $K$ does. Therefore, it is possible apriori
that relative homology groups may not be controlled by putting
conditions on the edges being contracted in $K$. Nevertheless, we show that for
every pair of subcomplexes $(L',L_0')$ in $K'$, there is a pair of
subcomplexes $(L,L_0)$ in $K$ such that $\gamma_{ab}(L,L_0) =
(L',L_0')$ and $\gamma_{ab*}: \homo_p(L,L_0) \rightarrow
\homo_p(L',L_0')$ is surjective if $ab$ satisfies the 
$(p-1)$-, and the $(p-2)$-link conditions in $K$.  This implies that no new
relative homology classes are created by such an edge contraction. 
In contrast, however, relative homology classes can be killed by edge
contractions even if it satisfies all link conditions. We develop these
results now.

The following result known as {\em five lemma}
in algebraic topology~\cite[Theorem 5.10]
{Rotman} lets us connect our result for homology to
relative homology groups.
\begin{theorem}[\cite{Rotman}.]
  Let $h:(L,L_0)\rightarrow(L',L_0')$ be a simplicial map. 
Let$f_{i}:\homo_i(L)\rightarrow \homo_p(L')$ and  
 $g_{i}:\homo_i(L_0)\rightarrow \homo_p(L_0')$ be
the homomorphisms induced by $h$ in the absolute
homology groups for $i=p, p-1$. 
The following statements hold:
\begin{enumerate}
\item[(i)] If $f_p$ and $g_{p-1}$ are surjective, and $f_{p-1}$ is injective,
then $ h_*:\homo_p(L,L_0)\rightarrow \homo_p(L',L_0') $ is surjective;
\item[(ii)] If $f_p$ and $g_{p-1}$ are injective, and $g_p$ is surjective,
then $ h_*:\homo_p(L,L_0)\rightarrow \homo_p(L',L_0') $ is injective;
\item[(iii)] If $f_i$ and $g_i$ are isomorphisms for $i=p, p-1$,
then $ h_*:\homo_p(L,L_0)\rightarrow \homo_p(L',L_0') $ is an isomorphism.
\end{enumerate}
\label{reltop-iso-thm}
\end{theorem}

If we take the map $h$ to be the restriction of the 
simplicial map $\gamma_{ab}$ to $L$ and its subcomplex $L_0$, we 
can use Theorem~\ref{homo-thm} to arrive at the following result:
\begin{theorem}
  Let $L\subseteq K$ be any pure subcomplex of dimension $p+1$ and
  $L_0\subset L$ be any of its pure subcomplexes of dimension $p$.
  Let $(L',L_0')=\gamma_{ab}(L,L_0)$. Then, the following hold:

\begin{enumerate}
\item[(i)] If $ab$ satisfies the $(p-2)$- and $(p-1)$-link conditions in $L_0$,\\ 
then $\gamma_{ab_*}: \homo_p(L,L_0) \rightarrow \homo_p(L',L_0')$ is surjective;
\item[(ii)] If $ab$ satisfies the $(p-1)$- and $p$-link conditions in $L_0$,\\ 
then $\gamma_{ab_*}: \homo_p(L,L_0) \rightarrow \homo_p(L',L_0')$ is injective;
\item[(iii)] If $ab$ satisfies the $(p-2)$- and $(p-1)$-, and 
$p$-link conditions in $L_0$,\\ 
then $\gamma_{ab_*}: \homo_p(L,L_0) \rightarrow \homo_p(L',L_0')$ is 
an isomorphism.
\end{enumerate}
\label{aux-topo-thm}
\end{theorem}
\begin{proof}
  We prove only (i) from which the proofs for (ii) and (iii) become obvious.  
  If $ab$ satisfies the $i$-link conditions for $i= p,p-2, p-1$
  in $L_0$, then it satisfies the same conditions for $L$ as well.
  Now apply Theorem~\ref{homo-thm} from previous section to
  conclude the following:
\begin{eqnarray*}
   &&\gamma_{ab*}:\homo_p(L)\rightarrow \homo_p(L') \mbox{ is surjective 
  if $ab$ satisfies $(p-1)$-link condition in $L_0$}\\
   &&\gamma_{ab*}:\homo_{p-1}(L_0)\rightarrow \homo_{p-1}(L_0') 
\mbox{ is surjective if $ab$ satisfies $(p-2)$-link condition in $L_0$}\\
   &&\gamma_{ab*}:\homo_{p-1}(L)\rightarrow \homo_{p-1}(L') 
\mbox{ is injective if $ab$ satisfies $(p-1)$-link condition in $L_0$}
\end{eqnarray*}
  Now apply Theorem~\ref{reltop-iso-thm}(i) to finish the proof of (i).
\end{proof}

It is important to notice that the required link conditions in
Theorem~\ref{aux-topo-thm} need to be satisfied in $L_0\subset L \subseteq K$.
It is not true that if $ab$ satisfies the $i$-link condition in
$K$, then it does so in $L$ and $L_0$. 
However, for $i=p-1,p-2$, we have the following observation
which allows
us to extend the results in Theorem~\ref{aux-topo-thm}(i) to
the case when $ab$ satisfies $(p-1)$- and $(p-2)$-link conditions in $K$.

%
\begin{proposition}
  Let $\gamma_{ab}: K\rightarrow K'$.
  Let $(L',L_0')$ be any pair of subcomplexes of dimensions $p+1$ and
  $p$, respectively, where $L' \subseteq K'$ and $L_0' \subset
  L'$. Let $L'$ contain the vertex $a$.  There exists a pair of
  subcomplexes $(L,L_0)$ of dimensions $p+1$ and $p$, respectively,
  where $L \subseteq K$ and $L_0 \subset L$ so that $\gamma_{ab}:
  (L,L_0) \rightarrow (L',L_0')$ and, for $i=p-1,p-2$, $ab$ satisfies the
  $i$-link condition for $L_0$ if it does so
  for $K$.  
\label{prop_LL}
\end{proposition}
\begin{proof}
  Consider the preimage of $L'$ under $\gamma_{ab}$ and take $L$ as
  its $(p+1)$-dimensional skeleton. Similarly take $L_0$ as the
  $p$-dimensional skeleton of the preimage of $L_0'$. For
  $i=p-1, p-2$, if there is a
  simplex $\sigma$ of dimension $(i-1)$,
  in $\Lk a\cap \Lk b$ in $L_0$, then the simplex
  $ab * \sigma$ has to be present in $L_0$ since this simplex is in $K$
  because $ab$ satisfies the $i$-link
  condition in $K$.  Therefore, $ab$ satisfies the same link
  condition in $L_0$ as well. 
\end{proof}

Notice that the above observation does not include the $p$-link condition.
Since the $p$-link condition may require a $(p+1)$-simplex,
the edge $ab$ may not satisfy the $p$-link condition
in the $p$-dimensional complex $L_0$, even if it does so in $K$. 
This is why we cannot extend Theorem~\ref{aux-topo-thm}(ii) and (iii)
in the result below which
relates the surjectivity of the relative homology groups 
with the link conditions
in the original input complex $K$ instead of a subcomplex. 
The proof of this theorem follows from applying Proposition~\ref{prop_LL} to
Theorem~\ref{aux-topo-thm}(i).

\begin{theorem}
  Let $L'\subseteq K'$ be any pure subcomplex of dimension $p+1$ and
$L_0'\subset L'$ be any of its pure subcomplex of dimension $p$,
where $K'=\gamma_{ab}(K)$. There exists a pair of pure
subcomplexes $L,L_0$ of dimensions $p+1$ and $p$, respectively with
$L\subseteq K$ and $L_0\subset L$ so that:\\

\noindent
If $ab$ satisfies $(p-2)$- and $(p-1)$-link conditions in $K$,
then $\gamma_{ab_*}: \homo_p(L,L_0) \rightarrow \homo_p(L',L_0')$ 
is surjective.
\label{main-topo-thm}
\end{theorem}

\subsection{Implications of Theorem~\ref{main-topo-thm} and 
Theorem~\ref{aux-topo-thm}}
\label{sec:implication}
\paragraph{Surjectivity.}
Notice that Theorem~\ref{main-topo-thm} is sufficient to claim that
edge contractions cannot create any new class in relative homology groups
if the edge $ab$ satisfies only $(p-1)$- and $(p-2)$-link conditions in $K$.
However, it may happen that the sujectivity does not respect torsions, that
is, the preimage of a torsion subgroup in $\homo_p(L',L_0')$ may not
have torsion in $\homo_p(L,L_0)$. This would mean that $(p-1)$- and $(p-2)$-link
conditions are not enough to guarantee that torsion
subgroups have sujection only
from the torsion subgroups of the source space. 
We are interested in link conditions that guarantee that
no new torsion class is created by an edge contraction because
of its connection to the problem of OHCP as discussed in the next section.

We present an example where a new relative torsion indeed appears in the
contracted complex even though the edge satisfies the $(p-1)$- and 
$(p-2)$-link conditions. Consider a sequence of triangles forming 
a M\"{o}bius strip. It is known that M\"{o}bius strip
has a torsion in $\homo_1$ relative to its boundary. Now remove one
triangle, say $abc$, and call the new complex $M$. 
Assume that the edge $ab$ was on the boundary of $M$.
The complex $M$, which is a M\"{o}bius strip
with one triangle removed, 
does not have any relative torsion. However, if we contract the edge $ab$, 
we eliminate the hole created by the removal of $abc$. The
resulting complex now is a M\"{o}bius strip and hence has a relative
torsion. Notice that, for $p=1$,  $ab$ satisfied $(p-1)$- 
and $(p-2)$-link conditions
vacuously though it did not satisfy the $p$-link condition.

The isomorphism result
in Theorem~\ref{aux-topo-thm}(iii) guarantees that no new relative
torsion appears if $ab$ satisfies all three link conditions,
namely $p$-, $(p-1)$-, and $(p-2)$-link conditions in $L_0$ which
cannot be guaranteed by requiring $ab$ to satisfy them in $K$.

In section~\ref{ssec-edgcontflagp}, by a graph theoretic approach,
we show that the $p$-link condition in $K$ alone
is sufficient to prevent the appearance of new relative torsions.
The above example illustrates that the $p$-link condition
is necessary for guaranteeing surjectivity in relative torsions. 

\paragraph{Injectivity.}
Unlike surjection, the injectivity implied by Theorem~\ref{aux-topo-thm}(ii)
cannot be used to claim that no relative homology class will
be killed by an edge $ab$ satisfying the $p$- and $(p-1)$-link conditions
in $K$. The reason is that for injectivity of $\gamma_{{ab}_*}$ 
we choose the pair $(L,L_0)$ in $K$ first and then consider its image
under $\gamma_{ab}$. So, even if $ab$ satisfies $p$- and 
$(p-1)$-link conditions in $K$, it may not do so in $L_0$ and
hence Theorem~\ref{aux-topo-thm}(ii) may not be applied. 
In fact, contracting edges satisfying all link conditions in $K$
can indeed kill a relative torsion. 
Figure~\ref{fig-contr_ex} illustrates
such an example.  The resulting complex after contracting edge $ab$ in
the $2$-complex on left is shown on right. The left complex minus the
triangles $abd$, $abm$ and the edge $ab$ forms a 15-triangle M\"obius
strip with a self-intersection at vertex $d$. This strip relative to
its boundary results in a relative torsion. On the right, the self
intersection expands to the {\em edge} $ad$, which causes the relative
torsion to disappear.
\begin{figure}[ht!]
\centering
\includegraphics[scale=0.8]{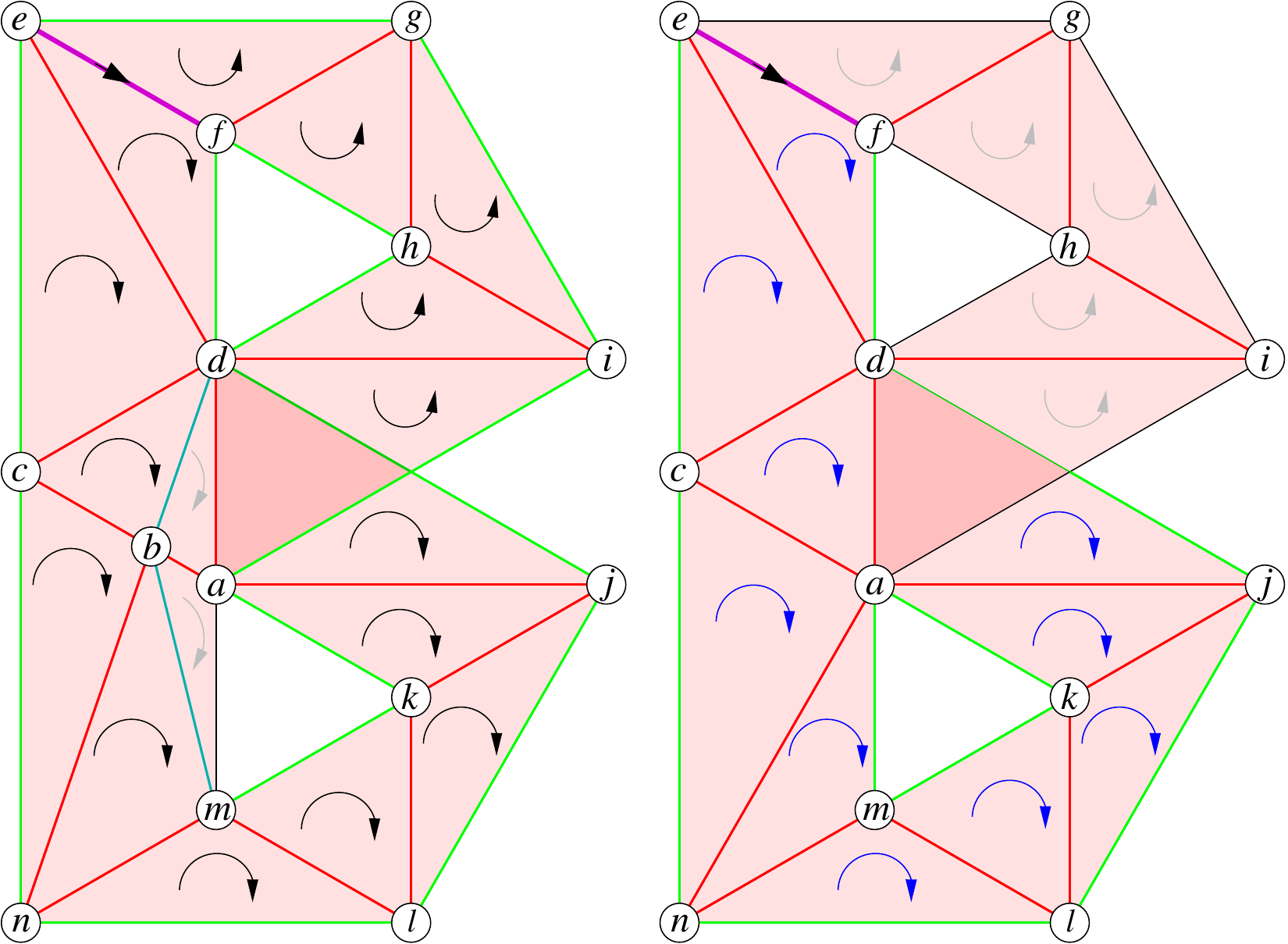}
\vspace{0.06in}
\caption{Contracting $ab$ destroys a relative homology class even
  though it satisfies the link condition. In the left complex, a
    M\"obius strip is formed by the triangles with orientations shown
    in black. This M\"obius strip does not exist in the right complex,
    since it self-intersects along edge $ad$ here. The input chain for
    the OHCP instance is edge $ef$ (shown in purple). Candidates for
    optimal homologous chains are shown in green and cyan on the left
    complex, and in green on the right complex. See Section
    \ref{ssec-exmpl} for more details.}
\label{fig-contr_ex}
\end{figure}

This example illustrates that an edge contraction may destroy a
relative torsion though no new relative torsion is generated thanks to
Theorem~\ref{main-topo-thm}(iii).  This is of course good news for OHCP
since such edge contractions can only take away obstructions to their
solutions and not introduce new ones.  We elaborate on this statement
now, and describe a concrete example illustrating the benefits of
such edge contractions to the efficient solution of OHCP instances.

\subsection{Relative homology preservation and the OHCP} 
\label{ssec-relhompresohcp}
Given an oriented simplicial complex $K$ of dimension $d$, and a
natural number $p$, $1 \le p \le d$, the \emph{p-boundary matrix} of
$K$, denoted $[\partial_{p}]$, is a matrix containing exactly one
column $j$ for each $p$-simplex $\sigma$ in $K$, and exactly one row
$i$ for each $(p-1)$-simplex $\tau$ in $K$.  If $\tau$ is not a face
of $\sigma$, then the entry in row $i$ and column $j$ is 0.  If $\tau$
is a face of $\sigma$, which we denote by $\tau \preceq \sigma$, then
this entry is $1$ if the orientation of $\tau$ agrees with the
orientation induced by $\sigma$ on $\tau$, and $-1$ otherwise.

Note that given a simplicial complex $K$ and a choice of orientations
for all simplices in $K$, $[\partial_{p}]$ is unique under row and
column permutations, and will generally be referenced as a single
matrix. We also need the notion of total unimodularity. A matrix $A$
is {\em totally unimodular}, or {\em TU}, if the determinant of each
square submatrix of $A$ is either $0$, $1$, or $-1$.  An immediately
necessary condition for $A$ to be TU is that each $a_{ij} \in \{0, \pm
1\}$.  The importance of TU matrices for integer programming is well
known -- see, for instance, the book by Schrijver \cite[Chapters
19-21]{Schrijver1986}. In particular, it is known that the {\em
  integer} linear program
\[
\min \; \vf^T \vx \quad
  \text{subject to} \quad A \vx = \vb, \; \vx \ge \vzero \text{
    and } \vx \in \Z^n
\]
for $A \in \Z^{m \times n}, \vb \in \Z^n$ can {\em always}, i.e., for
every $\vf \in \R^n$, be solved in polynomial time by solving its
linear programming {\em relaxation} (obtained by ignoring $\vx \in
\Z^n$) {\em if and only if $A$ is totally unimodular}.  The main
motivation for our discussion of totally unimodular matrices is the
following result in~\cite{DeHiKr2011}.
\begin{theorem}
\label{thm-tuiffnoreltor}
For a finite simplicial complex $K$ of dimension greater than $p$, the
boundary matrix $[\partial_{p+1}]$ is totally unimodular if and only
if $\homo_p\left(L,L_0\right)$ is torsion-free, for all pure
subcomplexes $L_0$, $L$ in $K$ of dimensions $p$ and $p+1$
respectively, where $L_0 \subset L$.
\end{theorem}
Following this result, instances of OHCP and related
problems (with homology over $\Z$) can be solved in polynomial time
when the simplicial complex $K$ is relative torsion-free. 
%

\subsection{An example where edge contraction helps to solve the OHCP} 
\label{ssec-exmpl} 
We present a small example which illustrates the effectiveness of edge
contractions on efficient solutions of OHCP.  Consider the simplicial
complexes $K$ on the left and $K'$ on the right in Figure
\ref{fig-contr_ex}. We obtain $K'$ from $K$ by contracting the edge
$ab$, which satisfies the link condition. 

We consider the following OHCP instance on $K$, and equivalently on
$K'$. All edges in red and the edge $ef$ in purple have weights of $1$
each. All edges in green and the thinner edges in black have weights
of $0.05$ each. The two edges $bd$ and $bm$ in $K$, drawn in cyan,
have weights of $0.10$ each. Orientations of all triangles and the
edge $ef$ are shown. Remaining edges could be oriented
arbitrarily. The input $1$-chain $\vc$ consists of the single edge
$ef$ with multiplier $1$.

The unique optimal solution to the OHCP LP is the chain $\vx$
consisting of the 15 edges shown in green and cyan in $K$, with
multipliers of $\pm0.5$, resulting in a total weight of $0.425$. $\vx$
is homologous to $\vc$ over $\R$, as their difference is the boundary
of the $2$-chain with multipliers of $-0.5$ for each of the $15$
triangles in the M\"obius strip, whose orientations are shown using
black arrows (triangles $abd$ and $abm$ have multipliers of $0$). But
the unique optimal homologous chain sought here is the $1$-chain
$\vx'$ shown in green in the right complex (this chain is identical in
$K$ and $K'$). Consisting of the 9 edges $ak, km, ma, ec, cn, nl, lj,
jd, df$ with multipliers $\pm1$ each, $\vx'$ has a total weight of
$0.45$. But to obtain this solution, we have to solve the OHCP model
as an {\em integer} program, instead of as a linear program.

Notice that $K'$ obtained from $K$ by contracting edge $ab$ is free of
relative torsion, and the weights of each edge in $K'$ is the same as
its weight in $K$. The unique optimal solution to the OHCP LP on $K'$
for the identical input chain is $\vx'$ itself. $\vx'$ is homologous
to $\vc$ as defined by the $2$-chain consisting of the triangles whose
orientations are shown using blue arrows, each with a multiplier of
$-1$. Hence we are able to solve the original OHCP instance using
linear rather than integer programming, after contracting one edge.

\bigskip
We show that
satisfying the $p$-link condition alone is sufficient to guarantee
that no new relative torsion is introduced.  Instead of proving the
result directly, we show that the $p$-link condition for the edge
being contracted preserves the total unimodularity of the boundary
matrix, which in turn guarantees the absence of relative torsion
thanks to Theorem~\ref{thm-tuiffnoreltor}.  In contrast to our earlier
approach, we use results from graph theory to arrive at this result.
 
\section{Bipartite Graphs and Boundary Matrices} 
\label{sec-bipgraphsbdymats}
Given a matrix $A$ with entries in \{0, $\pm$1\}, we associate with it
a weighted, undirected, bipartite graph $G=(V_1 \cup V_2, E)$ where
each row of $A$ corresponds to a node in $V_1$ and each column of $A$
corresponds to a node in $V_2$ \cite{Ca1965,CoRa1987}. Each nonzero
entry $a_{ij}$ in $A$ is associated with an edge connecting the nodes
in $V_1$ and $V_2$ corresponding to row $i$ and column $j$,
respectively, with a weight equal to $a_{ij}$. We call $G$ the
bipartite graph representation of $A$.

Some definitions from graph theory are central to our discussion.  For
a subgraph $S$ of $G$ containing a vertex $v$, we denote the number of
edges in $S$ incident to $v$ as the \emph{degree} of $v$ in $S$, or
$\deg_S(v)$.  A \emph{cycle} $Y$ is a connected subgraph of $G$ where
for each vertex $v \in Y, \ \deg_Y(v) = 2$.  We call a subgraph $C$ of
$G$ a {\em circuit} if for each vertex $v \in C, \deg_C(v) = 0 \bmod
2\,$. By this definition, it is possible that $v \in C$, but
$\deg_C(v) = 0$. However, we consider two circuits as equivalent if
they contain the same edge sets with the same weights for each edge.
This means, for example, that a circuit $C$ containing only vertices
of degree 0 in $C$ is equivalent to the empty subgraph.

\begin{definition}
  \label{def-bparity}
  A circuit $C$ in the weighted graph $G$ is {\em b-even} if the sum
  of the weights of the edges in $C$ is $0\bmod{4}$, and {\em b-odd}
  if the sum of the weights of the edges is $2\bmod{4}$. The quality
  of $C$ being b-even, b-odd, or neither is called the {\em b-parity}
  of $C$.
\end{definition}
This definition is equivalent to the definition of {\em even} and {\em
  odd} cycles given by Conforti and Rao \cite{CoRa1987}.  Given a
simplicial complex $K$, consider the $p$-graph $G_p(K)$ of $K$
constructed as follows. Each $p$- and $(p-1)$-simplex $\sigma\in K$
provides a dual vertex $\sigma^*$ in $G_p(K)$. We call the vertex
$\sigma^*$ a $p$- or $(p-1)$-vertex if $\sigma$ is a $p$- or
$(p-1)$-simplex, respectively.  There is an edge $\sigma^*\tau^*$ in
$G_p(K)$ if and only if the $(p-1)$-simplex $\sigma\in K$ has the
  $p$-simplex $\tau$ as a coface. The weight of $\sigma^*\tau^*$ is
$1$ or $-1$ depending on whether $\sigma$ and $\tau$ match
in orientation or not, respectively.  It is evident that
$G_p(K)$ is a weighted bipartite graph whose adjacency matrix is given
by $[\partial_p]$. This means the following proposition is almost
immediate.

\begin{proposition}
  \label{lem-revorntpresbparity}
  Reversing the orientation of any collection of $p$- and
  $(p-1)$-simplices in $K$ does not alter the b-parity of any circuit
  in $G_p(K)$.
\end{proposition}

Let $C$ be a circuit in a graph $G$.  A {\em chord} of $C$ is a single
edge of $G$ not in $C$ whose both end points are vertices in $C$. If
$C$ has no chord, it is called {\em chordless}, and to say $C$ is
induced is equivalent to saying $C$ is chordless. Using the
terminology given above, we now state without proof an important
result presented by Conforti and Rao \cite{CoRa1987}, who were in turn
motivated by the results of Camion \cite{Ca1965}.

\begin{theorem}
  \label{thm-tuiffnooddindcirct}
  For a matrix $A$ with entries in \{0, $\pm$1\} and its bipartite
  graph representation $G$, $A$ is totally unimodular if and only if
  $G$ contains no chordless b-odd circuit.
\end{theorem}

\noindent By combining Theorem~\ref{thm-tuiffnoreltor} and
Theorem~\ref{thm-tuiffnooddindcirct}, the following corollary is
immediate.
\begin{corollary}
  \label{cor-noreltoriffnooddincircuit}
  For a finite simplicial complex $K$ of dimension greater than $p$,
  the following results are equivalent.
  \begin{enumerate}
  \item $\homo_p\left(L,L_0\right)$ is torsion-free for all pure
    subcomplexes $L_0$, $L$ in $K$ of dimensions $p$ and $p+1$,
    respectively, where $L_0 \subset L$.
  \item The boundary matrix $[\partial_{p+1}]$ is totally unimodular.
  \item The bipartite graph representation $G_{p+1}(K)$ of
    $[\partial_{p+1}]$ contains no chordless b-odd circuit.
  \end{enumerate}
\end{corollary}

\subsection{Simplices and edge contraction} \label{ssec-simpedgcont}

Recall that each simplex in $K$ is mapped by an edge contraction
$\gamma_{ab}: K\rightarrow K'$ to a simplex in $K'$.  We categorize
simplices into three cases based on how they get mapped by
$\gamma_{ab}$. These cases are defined relative to the specific edge
$ab$ being contracted. We illustrate these cases in
Figure~\ref{fig-contr_case_exmpl}, and introduce several related
definitions below.
\begin{enumerate}

\item \label{Mirr_Case} 
  For each pair of simplices $\sigma, \sigma' \in K$ where $a \in
  \sigma, b\in \sigma'$, and $\sigma = ( \sigma' \setminus \{b\}) \cup
  \{a\}$, we have $\gamma_{ab}(\sigma) = \gamma_{ab}(\sigma') =
  \sigma$.  Then $\sigma$ and $\sigma'$ are \emph{mirror} simplices,
  and we say they are mirrors of each other.  Their duals are mirror
  vertices, and we say these vertices are mirrors of each other.

  In Figure~\ref{fig-contr_case_exmpl}, $a$ is the mirror of $b$, $ad$
  is the mirror of $bd$, $ae$ is the mirror of $be$, and $ade$ is the
  mirror of $bde$. Similarly, dual vertex $u_1$ is the mirror of
  $u_2$, $v_1$ is the mirror of $v_2$, and $w_1$ is the mirror of
  $w_2$.

\item \label{Coll_Case} 
  $\sigma \in K$ is \emph{collapsing} if $a, b \in \sigma$.  Its dual
  $\sigma^*$ is a \emph{collapsing} vertex.  Note if $\sigma$ is
  $p$-dimensional, then $\sigma$ has exactly one pair of ($p-1$)-faces
  that are mirrors of each other.  Then $\gamma_{ab}(\sigma) = \tau$,
  where $\tau$ is the unique mirror ($p-1$)-face of $\sigma$
  containing $a$.

  In Figure~\ref{fig-contr_case_exmpl}, $ab$, $abd$, $abe$, and $abde$
  are all collapsing simplices.  $v_3$, $w_3$, and $s$ are collapsing
  vertices.

\item \label{Inj_Case} 
  $\sigma \in K$ is \emph{injective} if neither of the above cases
  applies. We have $\gamma_{ab}(\sigma) = \sigma'$, and
  $\gamma_{ab}^{-1}(\sigma') = \sigma$.  If $b \notin \sigma$, then
  $\sigma' = \sigma$. If $b \in \sigma$, then $\sigma' = (\sigma
  \setminus \{b\}) \cup \{a\}$.
\end{enumerate}
\begin{figure}[ht!]
\centering
\includegraphics[scale=0.80]{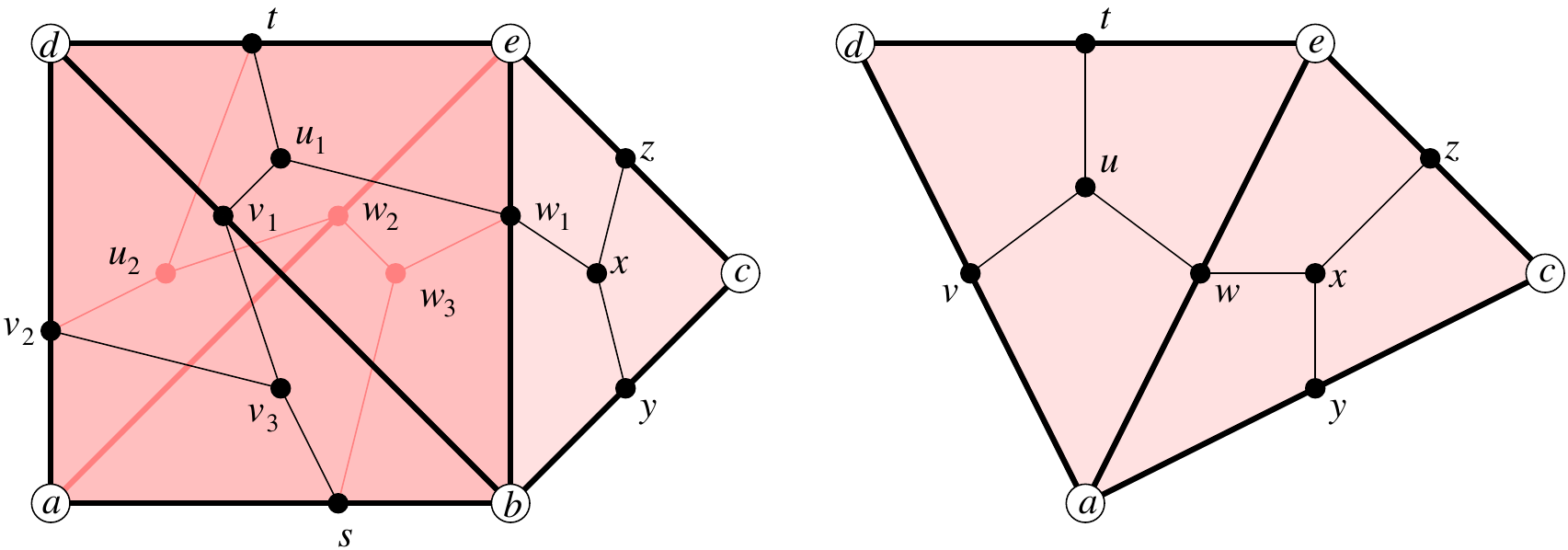}
\caption{The 2-graph of a complex, and the result after edge
  contraction $\gamma_{ab}$.}
\label{fig-contr_case_exmpl}
\end{figure}

\noindent We state a few more definitions related to the Mirror and
Collapsing cases.

\begin{definition}
  \label{def-Mirr_Edge}
  Let $\sigma_1 \in K$ be a ($p+1$)-simplex with mirror $\sigma_2 \in
  K$.  Note that $\sigma_1$ and $\sigma_2$ have exactly one common
  $p$-face $\xi$ that is an injective simplex.  Each other $p$-face
  $\tau_1 \preceq \sigma_1$ is a mirror of some $\tau_2 \preceq
  \sigma_2$, and vice-versa.  In any subgraph $S$ of the dual graph
  $G_{p+1}$ of $K$ that contains both edges $\sigma_1^*\tau_1^*$ and
  $\sigma_2^*\tau_2^*$, these two edges are \emph{mirror edges}, and
  are mirrors of each other.
  \noindent In Figure~\ref{fig-contr_case_exmpl}, 
   assuming $S$ is the entire $2$-graph, $tu_1$ is the mirror
  of $tu_2$, $u_1v_1$ is the mirror of $u_2v_2$, and $u_1w_1$ is the
  mirror of $u_2w_2$.
\end{definition}

\begin{definition}
  \label{def-mirrcon}
  For any $p$-simplices $\tau_1$ and $\tau_2$ that are mirrors of each
  other in $K$, the two unique edges of the dual graph $G_{p+1}$ of
  $K$ that directly connect $\tau_1$ and $\tau_2$ each to their common
  collapsing $(p+1)$-coface $\sigma$ are together called the
  \emph{mirror connection} between $\tau_1$ and $\tau_2$, and also
  between $\tau_1^*$ and $\tau_2^*$.
  \noindent There are two mirror connections in
  Figure~\ref{fig-contr_case_exmpl}.  The mirror connection between
  $v_1$ and $v_2$ are the two edges $v_3v_1$ and $v_3v_2$, and the
  mirror connection between $w_1$ and $w_2$ are $w_3w_1$ and $w_3w_2$.
\end{definition}

\begin{definition}
  \label{def-Coll_edge}
  Any edges incident to two vertices in the dual graph $G_{p+1}$ that
  are both collapsing are called \emph{collapsing edges}.
  \noindent In Figure~\ref{fig-contr_case_exmpl}, $sv_3$ and $sw_3$
  are collapsing edges.
\end{definition}

\section{Relative Torsion-Aware Edge Contractions}
\label{ssec-edgcontflagp}

We now present the theorem which states that contracting an edge in a
simplicial complex satisfying the $p$-link condition does not create
any new relative torsion in dimension $p$. In other words, if the
$(p+1)$-boundary matrix of a simplicial complex is totally unimodular
to start with, such an edge contraction will preserve this property.
Theorem~\ref{thm-circsurjfunc} below states the same thanks to
Corollary~\ref{cor-noreltoriffnooddincircuit}.

\begin{theorem}
  \label{thm-circsurjfunc}
  Let $K'=\gamma_{ab}(K)$ where $ab$ satisfies the $p$-link condition.
  Let $\crt$ be the set of circuits in graph $G=G_{p+1}(K)$, and
  $\crt'$ be the set of circuits in $G'=G_{p+1}(K')$. Let
  $\crtL\subseteq\crt$ be the set of all circuits of $G$ that contain
  a collapsing $p$-vertex.  Let $\crtM\subseteq\crt$ be the set of all
  circuits of $G$ that contain a pair of $(p+1)$-vertices that are
  mirrors of each other.  There exists a function $f: \crt \setminus
  (\crtM \cup \crtL)\rightarrow \crt'$ with the following properties.
  \begin{enumerate}
    \renewcommand{\theenumi}{{\bf P\arabic{enumi}}}
    \renewcommand{\labelenumi}{{\bf P\arabic{enumi}.}}
  \item \label{thm-circsurjfunc_surj} $f$ is surjective.
  \item \label{thm-circsurjfunc_bpar} $f$ preserves b-parity.
  \item \label{thm-circsurjfunc_chord} For each $C$ in the domain of
    $f$, if $C$ has a chord then so does $f(C)$.
  \end{enumerate}
\end{theorem}

\smallskip
\begin{proof}
  Before providing details of the proof, we refer the reader to
  Figure~\ref{fig-contr_case_exmpl} to visualize $\crtM$ and $\crtL$
  for $p=1$.  $\crtL$ is the set of all circuits that contain edges
  incident to the vertex $s$, which is the dual of the only collapsing
  $p$-simplex in this complex.  One such circuit is $\{sv_3,v_3v_1,
  v_1u_1, u_1w_1, w_1w_3, w_3s\}$.  $\crtM$ is the set of circuits
  which contain an edge incident to $u_1$, and also an edge incident
  to $u_2$.  One such circuit is $\{u_1t, tu_2, u_2v_2, v_2v_3,
  v_3v_1, v_1u_1\}$. Though there are several circuits in this graph,
  none are in the domain of $f$ as all of them belong to either
  $\crtM$ or $\crtL$.

  We break the proof up into several parts. We first define the
  function $f$ in Section \ref{sssec-deff}, and then show that $f$ is
  surjective (Property \ref{thm-circsurjfunc_surj}) in Section
  \ref{sssec-surjf}. Finally, we prove that $f$ preserves $b$-parity
  (Property \ref{thm-circsurjfunc_bpar}) in Section \ref{sssec-bparf}
  and that $f$ preserves chords of circuits (Property
  \ref{thm-circsurjfunc_chord}) in Section \ref{sssec-chordf}.

  \subsection{Definition of $f$} \label{sssec-deff}

  We define the function $f$ in the following way. For any $C \in
  \Dom(f)$, let $L=\{\sigma\, |\, \sigma^*\in C\}$.  For each vertex
  $\sigma^*$ of $C$, $f(\sigma^*)= \gamma_{ab}(\sigma)^*$.  The
  function $f$ maps each edge $\sigma^*\tau^*\in C$ to the edge
  $\gamma_{ab}(\sigma)^*\gamma_{ab}(\tau)^*$ if the edge is not a
  mirror connection, otherwise simply to the vertex
  $\gamma_{ab}(\sigma)^*=\gamma_{ab}(\tau)^*$.  Since $\gamma_{ab}$ is
  well defined, for any $C \in \Dom(f)$, if $f(C) = X$ and $f(C) = Y$,
  then $X = Y$.

  To show $f(C) \in \crt'$, first note that if all simplices of $L$
  are injective, this is almost trivially so since in these cases,
  $\gamma_{ab}$ is injective, preserves the dimension of simplices,
  and $\tau \in L$ is a face of $\sigma \in L$ if and only if
  $\gamma_{ab}(\tau)$ is a face of $\gamma_{ab}(\sigma)$.

  If $L$ does not contain any collapsing simplices, but contains a
  pair of mirror simplices $\tau_1$ and $\tau_2$, then since $C \notin
  \crtM$, $\tau_1$ and $\tau_2$ are $p$-simplices.  Furthermore, the
  proper cofaces of either $\tau_1$ or $\tau_2$ are injective because
  the only coface which cannot be such is a collapsing simplex.  Let
  $\tau = \gamma_{ab}(\tau_1) = \gamma_{ab}(\tau_2)$.  Since
  $\gamma_{ab}$ maps the proper cofaces of $\tau_1$ and $\tau_2$
  injectively, $\deg_{f(C)}(\tau^*) = \deg_{C}(\tau_1^*)
  +\deg_{C}(\tau_{2}^*)$.  Since both operands of this sum are even,
  the sum is even.  Therefore, $f(C) \in \crt'$.

  Now assume that $L$ contains a collapsing simplex $\sigma$.  Since
  $C \notin \crtL$, $\sigma$ is a $(p+1)$-simplex.  Furthermore, since
  all but two $p$-faces $\tau_1$ and $\tau_2$ of $\sigma$ are also
  collapsing, $\sigma^*$ must be of degree 2 in $C$.  Therefore, $C$
  must contain the mirror connections $\tau_1^*\sigma^*$ and
  $\tau_2^*\sigma^*$.  Note that $\sigma$ is the only $(p+1)$-simplex
  that is a coface of both $\tau_1$ and $\tau_2$.  The edge
  contraction $\gamma_{ab}$ maps other proper cofaces of $\tau_1$ and
  $\tau_2$ in $L$ injectively.  Both mirror connections
  $\tau_1^*\sigma^*$ and $\tau_2^*\sigma^*$ are mapped by $f$ to the
  vertex $\tau^*$ where $\tau = \gamma_{ab}(\tau_1) =
  \gamma_{ab}(\tau_2) = \gamma_{ab}(\sigma)$.  Therefore,
  $\deg_{f(C)}(\tau^*) = \deg_C(\tau_1^*) + \deg_C(\tau_2^*) - 2$.
  Again, since $\deg_C(\tau_1^*) + \deg_C(\tau_2^*)$ is even,
  $\deg_{f(C)}(\tau^*)$ is even.  Hence, $f(C) \in \crt'$.

  \subsubsection{Surjectivity of $f$} \label{sssec-surjf}

  We have shown that $f$ maps each element in its domain to an element
  in its codomain.  To show that $f$ is surjective (Property
  \ref{thm-circsurjfunc_surj}), we define a function that extends the
  domain of $f$ to all subgraphs of $G$.  Let $\subG$ be the set of
  all subgraphs of $G$, and $\subG'$ be the set of all subgraphs of
  $G'$.  Define the function $g : \subG \to \subG'$ as the extension
  of $f$ to this domain.  Hence it is possible that some vertex
  $\tau^* \in \Dom(g)$ is the dual of a collapsing $p$-simplex $\tau$.
  In this case we specify that $g$ maps $\tau^*$ and all edges of $G$
  incident to $\tau^*$, which all must be collapsing edges, to the
  empty subgraph.  Since $\gamma_{ab}$ is surjective, so is $g$.  To
  show $f$ is surjective, we will take an arbitrary $C' \in \crt'$ and
  show that its preimage under $g$, which we will denote $\subG_C$,
  must contain an element of $\Dom(f)$.

  If $\subG_C$ contains an element $S \notin \Dom(f)$, there are three
  type of edges that we may remove or add to $S$ without affecting
  $g(S)$, as detailed below.
  \begin{itemize}
    \renewcommand{\labelitemi}{$-$}
    \setlength{\parskip}{1pt}
  \item We may remove or add any collapsing edge.
  \item We may remove or add any edge that is a mirror connection in
    $G$.
  \item We may remove either one, but not both, of any two edges that
    are mirrors of each other.  For any edge in $S$ that is a mirror
    in $G$ but not in $S$, we may add its mirror to $S$.
  \end{itemize}

  \noindent Using these modifications, we describe a procedure to
  construct a $C_f \in \Dom(f) \cap \subG_C$ (see Figure
  \ref{fig-ProcConstCf}). Figures~\ref{fig-SandS1},~\ref{fig-S2andS3},
  and~\ref{fig-S4andS5} illustrate the steps of this procedure.  Graph
  vertices are labeled with the vertex labels of the dual simplices.
  In each figure, graph edges and vertices eliminated by the
  step shown are highlighted in red, and edges and vertices
  added by the step are highlighted in green.

  \begin{figure}[ht!]
    \label{proc-constCf}
    \framebox[6.5in]{
      \parbox{6.4in}{
        {\sc Procedure Construct $C_f$} \\
        \begin{tabular}{ll}
          {\tt Input:} & $S \in \subG_C \setminus \Dom(f)$ in graph $G$ such that $g(S) \in \crt'$.\\
          {\tt Output:} & $C_f \in \subG_C \cap \Dom(f)$ such that
          $g(C_f) = g(S)$.  \end{tabular}
        \begin{enumerate}
          \renewcommand{\theenumi}{\Roman{enumi}}
          \renewcommand{\labelenumi}{\Roman{enumi}.}
          \setlength{\parskip}{-1pt}
        \item \label{MkCf-remcoled} Remove any collapsing edges.
        \item \label{MkCf-addaSdMir} 
	For any edge $\sigma^* \tau^* \in S$ that is a mirror in $G$, 
        where $\sigma$ adjoins the
             contracting vertex $b$, replace it with its mirror in $G$ if
             this mirror is currently absent in $S$, or remove 
		$\sigma^* \tau^*$ otherwise.

        \item \label{MkCf-negMir} For each pair of vertices of odd
          degree that are mirrors of each other, negate the mirror
          connections connecting them, i.e., remove them if they are
          in the current subgraph, or add them if they are not.
        \item \label{MkCf-negSnged} For each pair of vertices of odd
          degree incident to a common single edge in $G$, negate this
          edge.
        \end{enumerate}
      }
    }
    \vspace*{-0.14in}
    \caption{ \label{fig-ProcConstCf} Procedure to construct a circuit
      in $\Dom(f)$ starting with a subgraph not in $\Dom(f)$.}
  \end{figure}
  \begin{figure}[hb!]
    \centering
    \includegraphics[scale=0.75]{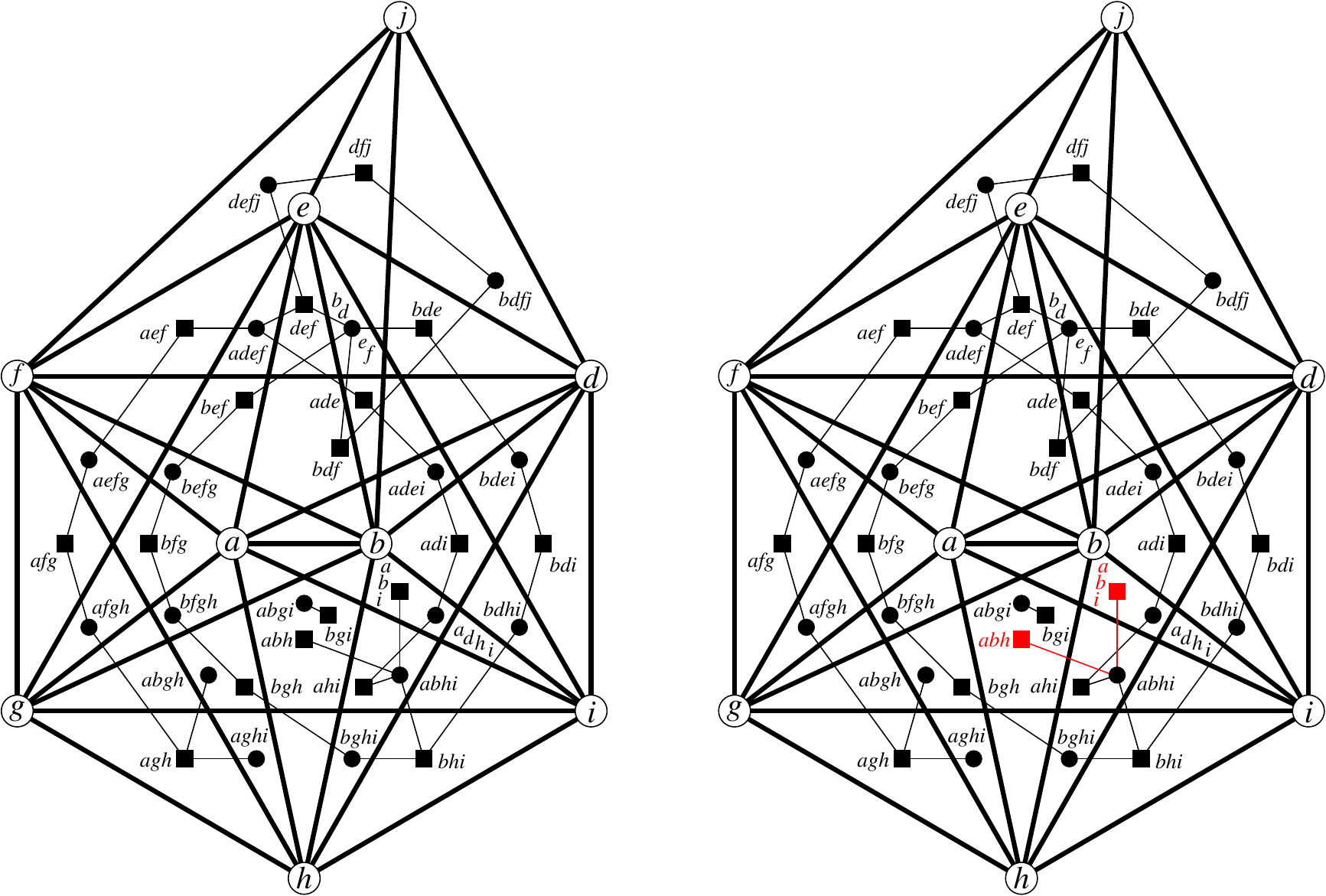}
    \caption{A subgraph $S$ of the 3-graph where $g(S) = C' \in \crt'$
      (left) and changes in Step I (right).}
    \label{fig-SandS1}
  \end{figure}

  We provide some details of the steps in the procedure now.  Denote
  the subgraph after Step \ref{MkCf-addaSdMir} as $S_2$.  Since we
  have removed any collapsing edges, the only edges remaining in $S_2$
  that $g$ does not map injectively are mirror connections in $G$.

  Since $C'$ is a circuit, the only vertices that are of odd degree in
  $S_2$ are vertices incident to a mirror connection in $G$.

  For any pair of mirror connections connecting two mirror vertices in
  $G$, there are three vertices involved -- two $p$-vertices
  $\tau_1^*$ and $\tau_2^*$ that are mirrors of each other in $G$, and
  a collapsing $(p+1)$-vertex $\sigma^*$.  $\deg_{S_2}(\sigma^*)$ is
  the number of edges of the mirror connection in $S_2$, and these are
  the edges $g$ removes by mapping them to the vertex $v = g(\tau_1^*)
  = g(\tau_2^*) = g(\sigma^*)$.  Because $C'$ is a circuit, and edges
  that are not a mirror connection are mapped injectively by $g$,
  $\,\deg_{S_2}(\sigma^*)$ is odd if and only if
  $\,\deg_{S_2}(\tau_1^*) + \deg_{S_2}(\tau_2^*)$ is odd. Therefore,
  among the three vertices $\tau_1^*, \tau_2^*, \sigma^*$, an even
  number of them must be of odd degree in $S_2$.  The graph on the
  right in Figure \ref{fig-S2andS3} illustrates this situation. Note
  that the vertices $ahi^*$ and $bhi^*$ have degrees $3$ and $1$ in
  $S_2$, respectively, and the vertex $abhi^*$ has degree $2$ in
  $S_2$.

  \begin{figure}[hb!]
    \centering
    \includegraphics[scale=0.75]{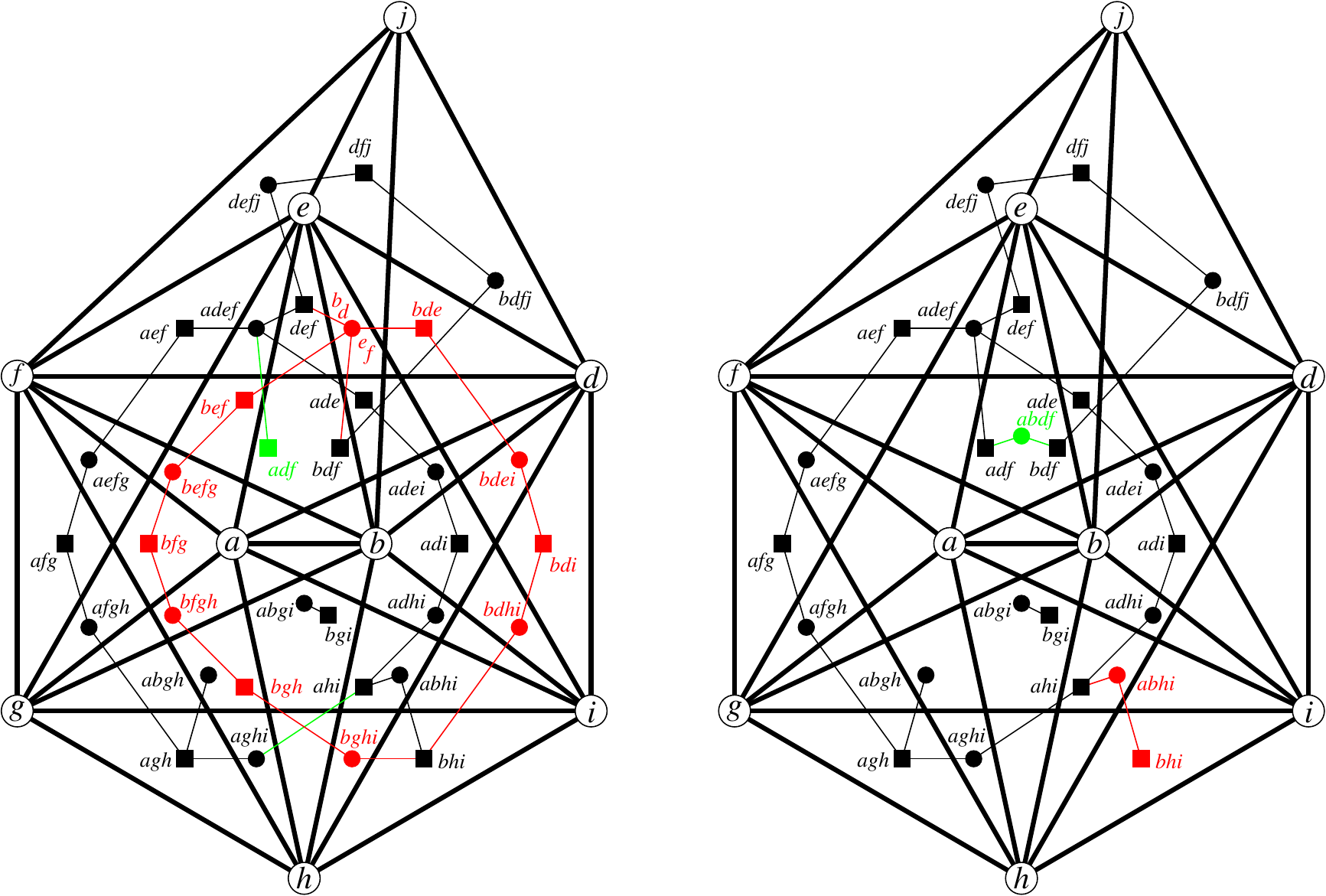}
    \caption{Changes in Step II (left) and in Step III (right).}
    \label{fig-S2andS3}
  \end{figure}
  
  If $\tau_1^*, \tau_2^*, \sigma^*$ are all of even degree, no further
  action is needed.  If this is not the case, but $\sigma^*$ is of
  even degree, then Step \ref{MkCf-negMir} will make $\tau_1^*$ and
  $\tau_2^*$ of even degree, and $\sigma^*$ will still be of even
  degree. Notice that since $ab$ satisfies the $p$-link condition, the
  mirror connections $\tau_1^*\sigma^*$ and $\tau_2^*\sigma^*$ must
  exist in $G$.  If $\sigma^*$ is of odd degree, then Step
  \ref{MkCf-negSnged} will convert $\sigma^*$ and whichever vertex
  between $\tau_1^*$ and $\tau_2^*$ is of odd degree to vertices of
  even degree.

  Steps \ref{MkCf-negMir} and \ref{MkCf-negSnged} ensure that every
  vertex in the subgraph have even degrees. The resulting subgraph
  $C_f$ is a circuit in $\subG_C$ that is not in $\crtM$ because we
  have removed at least one of the edges adjoining vertices that are
  mirror to each other in Step \ref{MkCf-addaSdMir}, and not in
  $\crtL$ because we have removed all collapsing edges in Step
  \ref{MkCf-remcoled}.  Therefore, $C_f \in \Dom(f)\ \cap\ \subG_C$.

  \begin{figure}[ht!]
    \centering
    \includegraphics[scale=0.75]{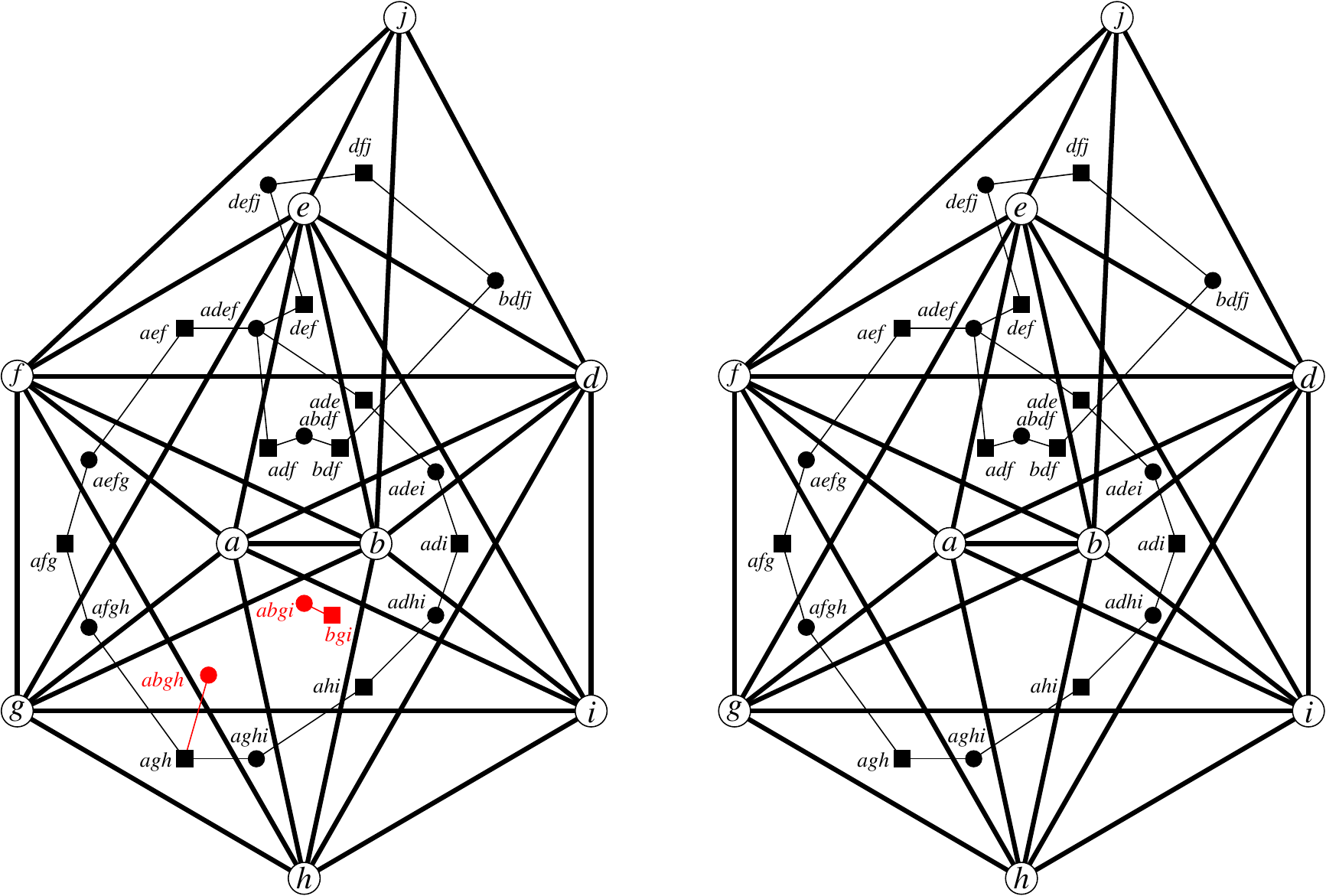}
    \caption{Changes in Step IV, and $C_f$.}
    \label{fig-S4andS5}
  \end{figure}

  \subsubsection{Preservation of $b$-Parity} \label{sssec-bparf}

  To show $f$ preserves $b$-parity (Property
  \ref{thm-circsurjfunc_bpar}), let $C \in \Dom(f)$, $C'=f(C)$, and
  $L=\{\sigma\,|\, \sigma^*\in C\}$.  We need to analyze only those
  vertices $\sigma^*$ of $C$ where $\sigma$ is collapsing, or is a
  mirror.  Since $C \in \Dom(f)$, we may restrict our attention to
  mirror $p$-vertices in $C$, or collapsing $(p+1)$-vertices in $C$.
  In fact, we may further restrict our attention to edges incident to
  $p$-vertices $\tau_1^*$ and $\tau_2^*$ that are mirrors and in $C$,
  because $C \in \Dom(f)$ implies any edges incident to a collapsing
  $(p+1)$-vertex is also incident to a $p$-vertex that is a mirror in
  $C$.

  We will divide the possibilities into four cases based on two
  criteria -- either the mirror connections connecting $\tau_1^*$ and
  $\tau_2^*$ are both in $C$ or both not in $C$, and either the
  orientations of $\tau_1$ and $\tau_2$ are consistent, or they are
  not.  Recall that if $\tau_1$ and $\tau_2$ have a unique common
  $(p-1)$-face $\xi$, $\tau_1$ and $\tau_2$ are consistently oriented
  if they induce opposite orientations on $\xi$.  We assume that
  unless otherwise specified, the orientations of simplices in $K'$
  are chosen so that for each edge $e$ in $C$, the weights of $e$ and
  $g(e)$ are the same.

  If neither edge of the mirror connection is in $C$, $f$ maps all
  edges of $C$ injectively to edges of $C'$. If the orientations of
  $\tau_1$ and $\tau_2$ are not consistent, then we may choose the
  ordering of vertices of $\tau = \gamma_{ab}(\tau_1)$ that determines
  its orientation to be the same as the ordering of the vertices of
  $\tau_1$.  Then each edge $e$ of $C'$ will have the same weight as
  $g^{-1}(e)$ in $C$.  If the orientations of $\tau_1$ and $\tau_2$
  are consistent, then we may still choose the order of vertices of
  $\tau$ in the same way as described above, but now each edge $e$
  incident to $\tau_2^*$ will have the opposite weight as $g(e)$.  The
  sum of the weights of edges incident to $\tau_2^*$ and the sum of
  the weights of their corresponding edges in $C'$ must differ by 2
  for each such edge $e$.  Since $C$ is a circuit, there must be an
  even number of such edges.  Therefore, this difference is $0 \bmod
  4$, and hence the $b$-parity of $C'$ is the same as that of
  $C$. We now discuss the remaining case when both mirror
  connections are in $C$ and the orientations of $\tau_1$ and $\tau_2$
  are not consistent. We illustrate both cases for $p=1$ in
  Figure \ref{fig-bparcases}.

  Consider the case when both mirror connections are in $C$ and the
  orientations of $\tau_1$ and $\tau_2$ are not consistent. Notice
  that the orientation of their common coface $\sigma$ must agree with
  the orientation of one of these mirrors and disagree with the
  orientation of the other mirror.  Hence, the sum of the weights of
  the mirror connections must be $0$. The function $f$ maps all other
  edges of $C$ injectively to edges of $C'$, and we may choose the
  orientation of $\tau$ as before. Therefore, the sum of the weights
  of edges in $C'$ equals the sum of the weights of edges in $C$.

  \begin{figure}[ht!]
    \centering
    \includegraphics[scale=0.8]{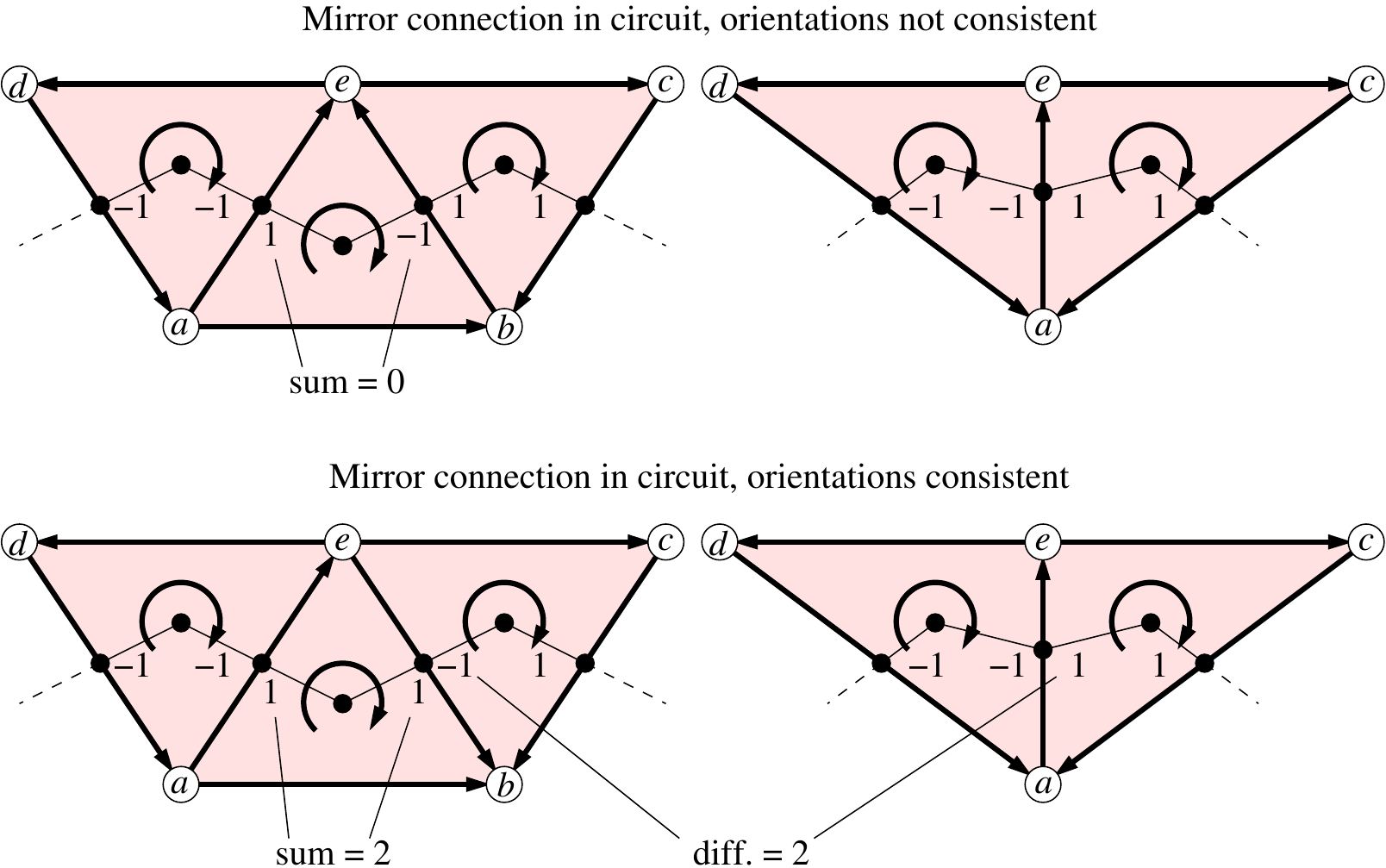}
    \caption{Examples of weights of edges of the 2-graph before (left)
      and after (right) edge contraction $\gamma_{ab}$.}
    \label{fig-bparcases}
  \end{figure}

  If the orientations of $\tau_1$ and $\tau_2$ are consistent, then
  the sum of the weights of the mirror connections must be $2 \bmod
  4$.  However, the number of edges mapped injectively by $g$ that are
  incident to $\tau_2^*$ must be odd.  Hence, the difference between
  the sum of weights of these edges and their corresponding edges in
  $C'$ must also be $2 \bmod 4$. It follows that the overall
  difference between the weights of edges in $C$ and $C'$ is $0 \bmod
  4$, and the $b$-parity of $C'$ equals that of $C$.  Thus, in all
  cases, $f$ preserves $b$-parity.  Figure~\ref{fig-bparcases}
  illustrates simple examples of the cases when the mirror connection
  is in $C$ for a $2$-graph.

  \subsubsection{Preservation of Chords} \label{sssec-chordf}

  To show property~\ref{thm-circsurjfunc_chord}, notice that each edge
  of $G$ is contained in a $(p+1)$-simplex in $K$.  If $C \in \Dom(f)$
  has a chord $h$, and $h$ adjoins a $(p+1)$-vertex $\sigma^*$, then
  there must be at least three $p$-vertices in $C$ -- at least two
  directly connected to $\sigma^*$ by edges in $C$, and also the
  $p$-vertex adjoining $h$.  Since each collapsing $(p+1)$-simplex has
  exactly two $p$-faces that are not collapsing, $\sigma$ cannot be a
  collapsing simplex.  Otherwise, $C \in \crtL$.

  If $L$ represents the set of all simplices in $K$ whose duals are in
  $C$, $\sigma$ cannot have its mirror in $L$ as $C \not\in \crtM$.
  Therefore, $\sigma$ is injective in $L$.
  Figure~\ref{fig-CandCprimechord} illustrates a chord in a 2-graph.
  Since $\gamma_{ab}$ maps $\sigma$ injectively, any edge of $G$ in
  $\sigma$ is in $C$ if and only if its image under $g$ is in $f(C)$.
  Therefore, $g(h)$ is not in $f(C)$, and for the vertex $\sigma^*$ of
  $h$, $g(\sigma^*)$ is in $f(C)$.  To show for the other vertex
  $\tau^*$ of $h$ that $g(\tau^*)$ is in $f(C)$, this could only be
  not true if all edges of $C$ incident to $\tau^*$ were not mapped to
  edges by $g$.  Since $C$ is a circuit, $\tau^*$ is incident to at
  least two edges.  But the only $p$-vertices incident to more than
  one edge in $G$ not mapped to an edge by $g$ are dual to collapsing
  $p$-simplices, which cannot hold for $\tau^*$ here since $C \notin
  \crtL$.
  \begin{figure}[ht!]
    \centering
    \includegraphics[scale=0.7]{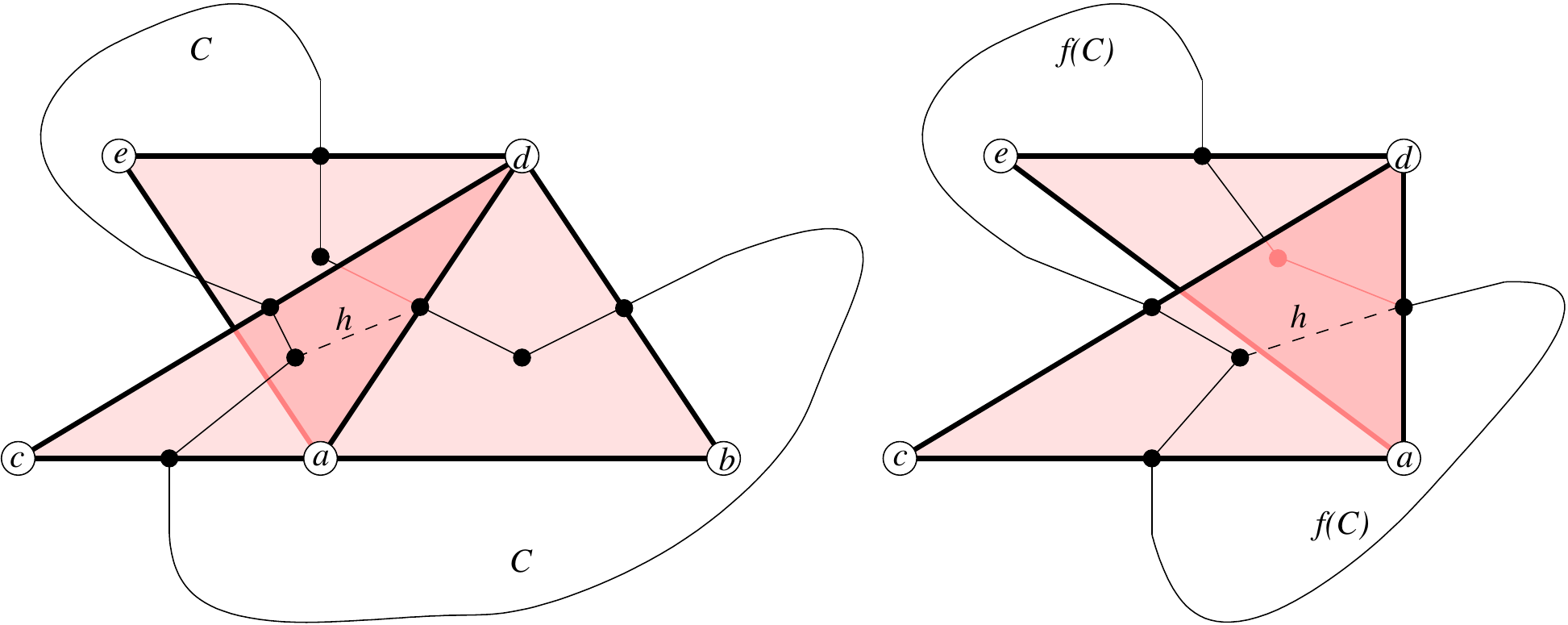}
    \caption{A chord $h$ of cycle $C$, and $C' = f(C)$.}
    \label{fig-CandCprimechord}
  \end{figure}
\end{proof}

\section{Discussion} \label{sec-disc}
We have presented several results that connect the local $p$-link
conditions to preserving various homological classes and torsions
under edge contractions.  We used both graph theoretic and algebraic
topological techniques to arrive at these results.
Our results on homology and relative homology groups may be
used to accelerate the computations of these topological
structures efficiently using edge contractions. More importantly,
we have laid down almost a complete picture of the relationship
of edge contractions in regards to simplicial homology.
It is not hard to see that the conditions in Theorem~\ref{homo-thm}
are not only sufficient, but are necessary for absolute homology
groups preservation. For relative homology, the scenario becomes
more subtle. The injectivity and hence isomorphism cannot be guaranteed
with any link condition in the original complex as our example
in Figure~\ref{fig-contr_ex} shows. The link conditions
in Theorem~\ref{main-topo-thm}
are sufficient for surjectivity; but are they necessary? This remains
open. For relative torsions, Theorem~\ref{thm-circsurjfunc}
and our example of punctured
M\"{o}bius strip in Section~\ref{sec:implication} show that the $p$-link
condition is sufficient and sometimes necessary for preventing new torsions,
and our example in Figure~\ref{fig-contr_ex} shows that no link
condition can ensure preserving existing relative torsions.

An open question is
whether there exist local conditions that also preserve the total
unimodularity property of the complex.  We know that total
unimodularity cannot be destroyed by an edge contraction if the
appropriate link condition holds.  But, we do not know how to preserve
the {\em absence} of total unimodularity with local conditions.  In
the Section~\ref{ssec-exmpl}, we discuss a concrete example where
linear programming fails to provide an optimal homologous chain, but
an edge contraction makes it amenable to such computation by
eliminating relative torsion.  An immediately relevant question is how
to approximate the optimal chain in the original complex using the one
in a reduced complex.

\section*{Acknowledgment} We acknowledge the support of the NSF grant
CCF-1064416. 

\bibliographystyle{plain}
\bibliography{homology}

\end{document}